\documentclass[letterpaper, 10 pt, conference]{ieeeconf}  

\IEEEoverridecommandlockouts                             

\usepackage{graphics}
\usepackage[english]{babel}
\usepackage{epsfig} 
\usepackage{mathptmx} 
\usepackage{times} 
\usepackage{amsmath,mathtools} 
\allowdisplaybreaks
\usepackage{array}
\usepackage{mdwmath}
\usepackage{amssymb} 
\newtheorem{theorem}{Theorem}

\DeclareMathOperator{\sinc}{sinc}
\newcommand{\comment}[1]{{ }}
\title{\LARGE \bf
Energy Optimal Control of a Harmonic Oscillator with a State Inequality Constraint
}

\author{Mi Zhou $^{1}$, Erik I Verriest $^{1}$, and Chaouki Abdallah $^{1}$
\thanks{$^{1}$  Mi Zhou, Erik I Verriest, and Chaouki Abdallah are with the School of Electrical and Computer Engineering, Georgia Institute of Technology, Atlanta, GA 30332.
        {\tt\small mzhou91@gatech.edu, erik.verriest@ece.gatech.edu, ctabdallah@gatech.edu}}
}

\begin{document}
\maketitle
\thispagestyle{empty}
\pagestyle{empty}
\begin{abstract}
In this article, the optimal control problem for a harmonic oscillator with an inequality constraint is considered.
The applied energy of the oscillator during a fixed final time period is used as the performance criterion.
The analytical solution with both small and large terminal time is found for a special case when the undriven oscillator system is initially at rest.
For other initial states of the Harmonic oscillator, the optimal solution is found to have three modes: wait-move, move-wait, and move-wait-move given a longer terminal time.
\end{abstract}

\section{INTRODUCTION}
{\it "Now, here, you see, it takes all the running you can do, to keep in the same place." from Alice in Wonderland, Lewis Carroll.}~\\
This sentence is used to describe Alice constantly running but remaining in place.
We found the same phenomenon in the optimal control of harmonic oscillators where the optimal behavior may involve ``remaining in place" for some time.
The details will be shown in Section \ref{sec:theory}.

The problem of minimum time optimal control of a pendulum is a classical problem  \cite{Pontryagin}.
When the control magnitude is bound, it is described by a second-order differential equation $\Ddot{x}+x = u$ where $x$ is the position with $|u| \leq 1$. The resulting optimal control $u$ is known as bang-bang control, i.e., $u=-\mathrm{sign}(\dot x)$.
Separating the domain $u=-1$ of the phase plane from the domain $u=+1$ are unit semicircle centered at points of the form $(2k+1,0)$ ($k\in \mathbb{Z}$) \cite{minidamp}.
In \cite{Russian}, the authors solved the optimal control problem for the harmonic oscillator with a fixed initial and final state under bounded control actions, and for two types of 
performance criteria.
Reference \cite{param} presented a solution to the minimum time control problem for a classical harmonic oscillator to reach a target energy from a given initial state by controlling its frequency.

Finding the analytical solutions to optimal control problems may not be possible in general.
However, there are many numerical methods to solve optimal control problems with/without constraints.
An introduction and reviews of numerical methods can be found in \cite{Bryson,survey,kelly}.

In this article, we model the optimal control problem of harmonic oscillators with state inequality and solve it using the indirect optimal control method based on Pontryagin's minimum (maximum) principle \cite{surveyMaximumP}.
The solutions are found to have three modes: wait-move (WM), move-wait (MW), and move-wait-move (MWM), which may be justified using human experience and physics.
To the best of our knowledge, this is the first work that such behavior has been found and explained.

This article is organized as follows: In Section \ref{sec:problem}, we formulate our problem and provide some preliminaries of optimal control with global inequalities.
In Section \ref{sec:theory}, we present the analytical solution to this formulated optimal control problem.
Section \ref{sec:simulation} presents the simulation results to illustrate our analytical solutions.
Finally, we conclude our article in Section \ref{sec:conclusion}.

\section{Problem description}\label{sec:problem}
A general undamped harmonic oscillator has the form $m\Ddot{x}+kx=u$.
Figure \ref{fig:hc} is a harmonic oscillator with mass $m$, stiffness $k$, and external excitation $u$.
Without loss of generality, 
we let $m=1$, $k=1$ in our problem.
We consider the following linear state space realization for the undamped  harmonic oscillator with fixed oscillation frequency:
\begin{align} \label{eqn:system}
\dot x_1 &= x_2 \nonumber\\ 
\dot x_2 &= -x_1+u,
\end{align}
where $x_1 \in \mathbb{R}^1$ is the position, $x_2 \in \mathbb{R}^1$ is its velocity, and $u \in \mathbb{R}^1$ is the control input.
The performance criterion is the applied energy
\begin{align}\label{eqn:perfo}
J= \frac{1}{2}\int_0^T u^2 \mathrm{d}t,
\end{align}
with $T$ the fixed terminal time.
We consider the  state inequality constraint
\begin{align}\label{eqn:constr}
x_2 \geq 0.
\end{align}
which means only forward motion is allowed.
We would like to find the optimal control law $u^*(t)$ such that the criterion \eqref{eqn:perfo} is minimized and the dynamics \eqref{eqn:system} are satisfied with the inequality constraint \eqref{eqn:constr}.
\begin{figure}[!htp]
    \centering
    \includegraphics[width=\linewidth]{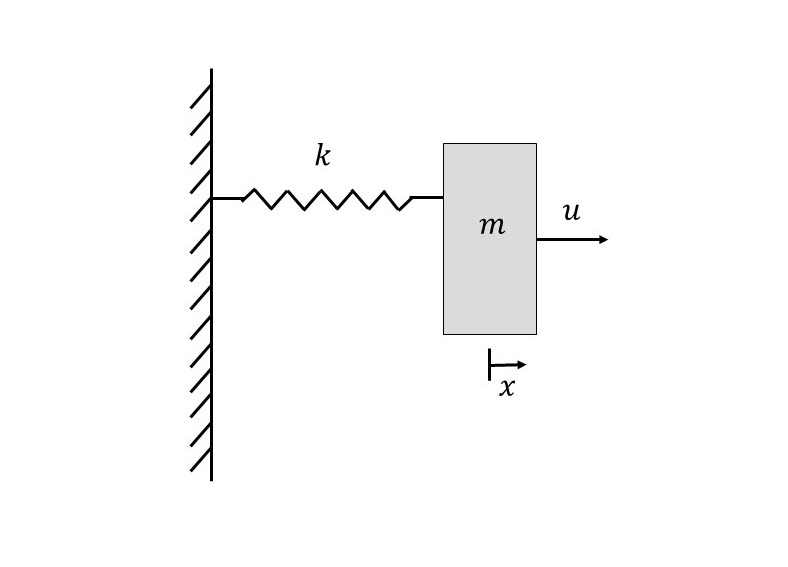}
    \caption{Mechanical harmonic oscillator.}
    \label{fig:hc}
\end{figure}
\subsection{Optimal control with inequality constraints}
For a general optimal control problem 
$$\dot{x} = f(x,u,t)$$
with global inequality constraint $w(x,u,t) \leq 0$, we can construct the Hamiltonian function 
\begin{equation*}
H = L(x,u) + \lambda^\top f(x,u,t)+\mu^\top w(x,u,t),
\end{equation*}   
where $\mu = [\mu_1, \mu_2, \cdots, \mu_r]$.
\begin{theorem}[\cite{optimalcontrolbook}] \label{thm:optimality}
The multiplier $\mu_i$, $i=1,2,\cdots,r$ must satisfy the complimentary slackness condition:
\begin{align*}
\mu_i \left\{\begin{matrix}
\geq 0, \quad w_i(x,u,t) = 0\\ 
=0,\quad w_i(x,u,t)<0
\end{matrix}\right..
\end{align*}
The necessary conditions for $u^*$ (corresponding $x^*$) to be an optimal solution are that there exist $\lambda$, $\mu$ which satisfy the following conditions:
\begin{enumerate}
    \item $\dot x^* = f(x^*, y^*, t)$, $x^*(0)=x_0$, $x^*(T)=x_f$.
    \item Pontryagin minimum condition:\\ $H[x^*(t), u^*(t),\lambda(t), t] \leq H[x^*(t), u(t), \lambda(t), t]$.
    \item The Euler-Lagrange equation:
    \begin{align*}
      \dot \lambda  = \left(-\frac{\partial L}{\partial x} -\lambda^\top \frac{\partial f}{\partial x} +\mu^\top \frac{\partial w}{\partial x}\right)\mid_{x=x^*}.
    \end{align*}
\end{enumerate}
\end{theorem}

The solution of an optimal control problem with global inequality constraints is continuous but may not be a differentiable function.

In the following section, we will use Theorem \ref{thm:optimality} to analyze the solution of the defined optimal control problem Eqn. \eqref{eqn:system}, \eqref{eqn:perfo}, \eqref{eqn:constr}.

\section{Solution of the problem: Theoretical analysis}\label{sec:theory}
The proposed system \eqref{eqn:system} models a mechanical spring system.
Intuitively, when there is no input, and if the initial displacement is greater than zero, the spring is in a state of stretching. 
The system will then try to go to the state equilibrium $(0,0)$ thus requiring the velocity to be less than zero at the beginning and then behave as a harmonic oscillator.
If the spring is initially compressed, the initial velocity will be positive then show the tendency of harmonic oscillation.

To solve the optimal control problem for such a system, we first construct the following Hamiltonian
\begin{align*}
H = \frac{1}{2}u^2 +\lambda_1 x_2 +\lambda_2(-x_1+u)-\mu x_2.
\end{align*}
Using the optimality conditions derived above, we have
\begin{align} \label{eqn:OC}
u &= -\lambda_2 \nonumber\\
\dot \lambda_1 &= \lambda_2 \nonumber\\
\dot \lambda_2 &= -\lambda_1 + \mu
\end{align}
where $\mu=0$ if the constraint is not activated, otherwise it must satisfy $\mu\geq 0$.

Solving \eqref{eqn:system} and \eqref{eqn:OC} with the constraint not activated (i.e., $\mu=0$), we have
{\small 
\begin{align}\label{eqn:ori_sol}
x_1(t) &= C_1 \cos t +C_2 \sin t +C_3 t \cos t +C_4 t\sin t \\
\nonumber x_2(t) &= (C_4-C_1)\sin t +(C_2+C_3)\cos t -C_3 t \sin t +C_4 t \cos t \\
\nonumber u(t) &= -2C_3 \sin t +2 C_4\cos t
\end{align}
}%
where $C_1, C_2, C_3, C_4$ are constants.

Denote the initial state as $(s,0)$ and the terminal state as $(x_f,0)$.
We then analyze the solution to the above problem from four different scenarios: 
\begin{enumerate}
    \item 
$s=0$ and $x_f>0$;
\item
$0<s<x_f$;
\item $s<0<x_f$; 
\item 
$s<x_f<0$.
\end{enumerate}
The scenario that $x_f<s$ can not happen because the inequality constraint implies that the spring can not move backward.
The terminal time is denoted as $T$.
\subsection{Initial state $(0,0)$ and $x_f>0$} \label{subsec:s1}
In the first scenario, we consider the initial state $(0,0)$  and the terminal position $x_f>s=0$.
The initial state $(0,0)$ means that the harmonic oscillator is in the natural state at the beginning, neither stretching nor compressing, with no velocity either.

Substituting the initial state $(x_1(0), x_2(0)) = (0,0)$ into Eqn. \eqref{eqn:ori_sol}, we obtain
\begin{align*}
C_1 &= 0\\
C_2+C_3 &= 0.
\end{align*}
Henceforth,
\begin{align}
x_1(t) &= C_2\sin t -C_2 t \cos t+C_4 t \sin t
\end{align}
and 
\begin{align} \label{eqn:x2}
x_2(t) &= C_4 \sin t +C_2 t \sin t +C_4 t \cos t.
\end{align}
Using the terminal state $(x_1(T), x_2(T))=(x_f,0)$, we can formulate the following equality with respect to the unknown constants $C_2$ and $C_4$,
\begin{align*}
\underbrace{\begin{bmatrix}
\sin T-T\cos T & T\sin T\\ 
T \sin T &T\cos T+\sin T 
\end{bmatrix}}_{M}
\begin{bmatrix}
C_2\\ 
C_4
\end{bmatrix} = \begin{bmatrix}
x_f\\ 
0
\end{bmatrix}.
\end{align*}
It is obvious that $\det(M) = \sin^2T-T^2 <0$ for all $T>0$.
Thus solving the above equations, we can obtain
\begin{align}
C_2 &= -\frac{x_f(\sin T + T\cos T)}{T^2-\sin^2T} \\
C_4 &= \frac{x_f T  \sin T}{T^2-\sin^2T}.
\end{align}
Substituting the value of $C_2$ and $C_4$ into Eqn. \eqref{eqn:x2}, we have
\begin{align} \label{eqn:finalx2}
x_2(t) = \beta T t \sin(T-t) + \beta (T-t)\sin T \sin(t)
\end{align}
and
\begin{align} \label{eqn:optimalcontrol}
u(t) = -2\beta \sin T \sin t+2\beta T \sin (T-t)
\end{align}
where $\beta = \frac{x_f}{T^2-\sin^2 T} >0$.
Analyzing the Eqn. \eqref{eqn:finalx2}, we find that there exists a $T^*$ when $T\leq T^*$, $x_2(t) \geq 0$.
\begin{theorem}
The boundary of the terminal time $T$ for the state becoming nonsmooth is
\begin{align*}
T^* = \pi,
\end{align*}
\end{theorem}
\begin{proof}
First, we re-organize $x_2(t)$ as follows:
\begin{align*}
x_2(t) &= \beta T t (T-t) \left [ \frac{\sin(T-t)}{T-t}  +\frac{\sin T}{T}\frac{\sin t}{t}  \right ] \\
&= \underbrace{\beta T t (T-t)}_{\geq 0} [\sinc(T-t)+ \sinc(T) \sinc(t)]
\end{align*}
where we define $\sinc(x)=\frac{\sin x}{x}$.
When $T \in (0,\pi]$, $\sinc(T) \geq 0,\; \sinc(t) \geq 0,\; \sinc(T-t)\geq 0$. Thus, $x_2(t)\geq 0$.
Suppose $T=\pi+\epsilon$ where the small perturbation $\epsilon>0$ is the new bound.
Substitute it into 
\begin{align*}
\gamma(t) &= \sinc(T-t)+\sinc(T)\sinc t \\
&= \frac{-\cos(\epsilon -t)}{\pi+\epsilon-t}-\frac{\cos\epsilon}{\pi+\epsilon}\frac{\sin t}{t}. 
\end{align*}
We then find that $\gamma(0) = -\frac{\cos\epsilon}{\pi+\epsilon} -\frac{\cos \epsilon}{\pi +\epsilon}<0$ since $\lim_{t\rightarrow 0}\frac{\sin t}{t} = 1$.
Thus $T^*$ must be $\pi$.
\end{proof}

If $T$ is chosen to satisfy that $x_2(t) \geq 0$ for $t\in [0,T]$, then Eqn. \eqref{eqn:finalx2} is the analytical solution.
When the terminal time $T>\pi$, the above analysis fails.
Looking at the system dynamics, we find that there is no cost to keep the state stagnating at the initial state.
That being said, the optimal solution in this case is ``WM'':
\begin{align} \label{eqn:x1_general}
x_1(t) = \left\{\begin{matrix}
0, \quad t\in [0, \tau]\\ 
x_1^c(t), \quad t\in [\tau,T]
\end{matrix}\right.
\end{align}
and 
\begin{align} \label{eqn:x2_general}
x_2(t) = \left\{\begin{matrix}
0, \quad t\in [0, \tau]\\ 
x_2^c(t), \quad t\in [\tau,T]
\end{matrix}\right.
\end{align}
where $\tau = T-\pi$, $x_1^c(t)$ and $x_2^c(t)$ are continuously differentiable functions dependent on time.
The corresponding optimal controller has the form:
\begin{align}
u(t) = \left\{\begin{matrix}
0, \quad t \in [0,\pi]\\ 
u^c(t),\quad t \in [\pi, T]
\end{matrix}\right. .
\end{align}

Substituting the boundary condition $(x_1(\tau),x_2(\tau) )= (0,0)$, we thus have
{\small 
\begin{align*}
C_1 \cos \tau +C_2 \sin \tau+C_3\tau \cos \tau+C_4 \tau \sin \tau &= 0\\
(C_4-C_1)\sin \tau+(C_2+C_3)\cos \tau
-C_3 \tau \sin \tau +C_4 \tau\cos \tau &= 0\\
C_1 \cos T+C_2 \sin T +C_3 T \cos T +C_4 T \sin T &= x_f \\
(C_4-C_1)\sin T +(C_2+C_3)\cos T-C_3 T\sin T+C_4 T \cos T &= 0.
\end{align*}
}%
Substituting $\tau = T-\pi$ and simplifying the above equation, we have
{\small 
\begin{align*}
&\begin{bmatrix}
-\cos T &-\sin T  &-(T-\pi)\cos T  & -(T-\pi)\sin T\\ 
\sin T & -\cos T & (T-\pi)\sin T-\cos T &-\sin T-(T-\pi)\cos T \\ 
\cos T &\sin T  & T\cos T &T\sin T \\ 
-\sin T & \cos T & \cos T-T \sin T &\sin T+T \cos T 
\end{bmatrix}\\
& \cdot \begin{bmatrix}
C_1\\ 
C_2\\ 
C_3\\ 
C_4
\end{bmatrix}=\begin{bmatrix}
0\\ 
0\\ 
x_f\\ 
0
\end{bmatrix}.
\end{align*}
}%
By solving the above linear equations, we obtain
\begin{align*}
C_1 &= \frac{x_f \sin T+\pi x_f \cos T-Tx_f \cos T}{\pi}\\
C_2 &= - \frac{x_f (\cos T-\pi\sin T+T \sin T)}{\pi}\\
C_3 &= \frac{x_f\cos T}{\pi}\\
C_4 &= \frac{x_f\sin T}{\pi}.
\end{align*}
Thus the optimal controller for the time interval $[\pi, T]$ is
\begin{align*}
u^c(t) = \frac{2x_f}{\pi}\sin (T-t)
\end{align*}
which makes the optimal cost
\begin{align*}
J^* = \frac{x_f^2}{\pi},
\end{align*}
which implies that the optimal cost is only a function of $x_f$.
\subsection{Initial state $(s,0)$}
Here Eqn. \eqref{eqn:ori_sol} still holds.
The key point here is to confirm the constant $C_1 \sim C_4$ using known boundary conditions.
We then study the solution for the initial arbitrary state on the constraint.
When $T$ is small, we substitute the initial state $(s,0)$ into Eqn. \eqref{eqn:ori_sol} and obtain
\begin{align*}
M \begin{bmatrix}
C_2 \\ 
C_4 
\end{bmatrix} = \begin{bmatrix}
x_f-s\cos T\\ 
s \sin T
\end{bmatrix}
\end{align*}
which gives
\begin{align*}
C_1 &= s \\
C_2 &= \frac{s(T+\cos T \sin T)-x_f (\sin T+T\cos T)}{T^2-\sin^2 T} \\
C_3 &= -C_2 \\
C_4 &= \frac{x_f T \sin T -s \sin^2 T}{T^2-\sin^2 T}.
\end{align*}
Thus,
\begin{align} \label{eqn:as_anys}
\nonumber x_1(t) = s \cos t +C_2 (\sin t -t \cos t)+C_4 t \sin t\\
x_2(t) = C_4 (\sin t +t \cos t) -s \sin t +C_2 t \sin t.
\end{align}

When $T$ is large, the constraint will be activated at some time interval.
Here, we divide it into the following three cases.
\subsubsection{Case 1: $0<s<x_f$}
When $T$ is large, the stay will be in the initial state (i.e., ``WM'').
Similar to the case when $s=0$, we assume $\tau$ as the staying time.
Thus, the solutions of $x_1$ and $x_2$ have the same form as Eqn. \eqref{eqn:x1_general} and Eqn. \eqref{eqn:x2_general} respectively.
The controller during $[0,\tau]$ is the value such that
\begin{align*}
\dot x_2 = -x_1+u =0 \Rightarrow -s+u = 0\Rightarrow u = s
\end{align*}
which implies that during this time interval, with the constant control input $u=s$, we can keep the state stay at the initial state.
As we know, we have five unknowns, i.e., $C_1\sim C_4$ and $\tau$.
Using the fact that $u(\tau)=s$ and the boundary conditions of states, we have the following five constraints:
{\small 
\begin{align*}
&C_1 \cos \tau +C_2 \sin \tau+C_3\tau \cos \tau+C_4 \tau \sin \tau = s\\
&(C_4-C_1)\sin \tau+(C_2+C_3)\cos \tau
-C_3 \tau \sin \tau +C_4 \tau\cos \tau = 0\\
&C_1 \cos T+C_2 \sin T +C_3 T \cos T +C_4 T \sin T = x_f \\
&(C_4-C_1)\sin T +(C_2+C_3)\cos T-C_3 T\sin T+C_4 T \cos T = 0 \\
&-2C_3 \sin \tau+2C_4 \cos \tau = s.
\end{align*}
}%
which allows us to solve for the constants.
The cost function is
\begin{align*}
J &= \int_0^\tau \frac{1}{2} s^2 \mathrm{d}t + \int_\tau^T \frac{1}{2}u^2 \mathrm{d}t \\
&= \frac{1}{2}s^2 \tau + \frac{1}{2}\int_\tau^T (-2C_3\sin t +2C_4 \cos t)^2 \mathrm{d}t.
\end{align*}
\subsubsection{Case 2: $s<0<x_f$}
When $T$ is large, the stay will be in the middle because it needs less energy to stay there (i.e., ``MWM'').
We denote the staying time interval as $[\tau_1, \tau_2]$.
During $[\tau_1, \tau_2]$,
\begin{align*}
\dot x_1 = x_2 = 0 \\
\dot x_2 = -x_1+u = 0,
\end{align*}
which gives
\begin{align}
x_{1s}=x_1(t) = u, \quad \forall  t \in [\tau_1, \tau_2] \\
x_{2s}=x_2(t) = 0 , \quad \forall t \in [\tau_1, \tau_2]
\end{align}
where $(x_{1s}, x_{2s})$ is a constant staying state.
Thus the solution has the following form:
\begin{align*}
u(t) = \left\{\begin{matrix}
u^c(t),\quad \forall t \in [0,\tau_1]\\ 
x_{1s}, \quad \forall t \in [\tau_1, \tau_2]\\ 
u^p(t), \quad \forall t \in [\tau_2, T]
\end{matrix}\right.
\end{align*}
where $u^c(t)$, $u^p(t)$ has the same formula as in \eqref{eqn:ori_sol},
and 
\begin{align}
x_1(t) = \left\{\begin{matrix}
x_1^c(t), \quad t \in [0, \tau_1]\\ 
x_{1s}, \quad t\in [\tau_1, \tau_2]\\ 
x_1^p(t), \quad t \in [\tau_2, T]
\end{matrix}\right.
\end{align}
and
\begin{align}
x_2(t) = \left\{\begin{matrix}
x_2^c(t), \quad t \in [0, \tau_1]\\ 
x_{2s}, \quad t\in [\tau_1, \tau_2]\\ 
x_2^p(t), \quad t \in [\tau_2, T]
\end{matrix}\right.
\end{align}
where $x_1^c(t), \; x_1^p(t),\; x_2^c(t),\; x_2^p(t)$ satisfying the form of Eqn. \eqref{eqn:ori_sol} are some continuously differentiable functions dependent on time $t$.
There are 11 unknowns to be solved, i.e., $C_1\sim C_4$ (parameters for $x_{1,2}^c(t)$), $C_1'\sim C_4'$ (parameters for $x_{1,2}^p(t)$), $\tau_1$, $\tau_2$, and $x_{1s}$ while the known boundary conditions are
\begin{align*}
\left\{\begin{matrix}
x_1^c(0)&=s\\ 
x_2^c(0)&=0\\ 
x_1^c(\tau_1)&=x_{1s}\\ 
x_2^c(\tau_1)&=0\\ 
x_1^p(\tau_1) &= x_{1s}\\
x_2^p(\tau_2)&=0\\ 
x_1^p(T)&=x_f \\
x_2^p(T)&=0 \\
u^c(\tau_1) &= x_{1s}\\
u^p(\tau_2) &= x_{1s}
\end{matrix}\right.
\end{align*}
Combined with optimizing the cost function, 
\begin{align*}
J = \int_0^{\tau_1} \frac{1}{2}u^2 \mathrm{d}t + \frac{1}{2}x_{1s}u^2 (\tau_2-\tau_1)+\int_{\tau_2}^T \frac{1}{2}u^2 \mathrm{d}t
\end{align*}
we can finally obtain the solutions.
\subsubsection{Case 3: $s<x_f<0$}
When $T$ is large, the stay will be in the terminal state (i.e., ``MW'').
Thus the solution has the following form:
\begin{align}
x_1(t) = \left\{\begin{matrix}
x_1^c(t), \quad t\in [0,\tau]\\ 
x_f, \quad t \in [\tau, T]
\end{matrix}\right.
\end{align}
and
\begin{align}
x_2(t) = \left\{\begin{matrix}
x_2^c(t), \quad t\in [0,\tau]\\ 
0, \quad t \in [\tau, T]
\end{matrix}\right..
\end{align}
This means that when $t\in [\tau, T]$,
\begin{align*}
\dot x_1(t) &= x_2(t) = 0 \\
\dot x_2(t) &= -x_1(t) + u =0 \Rightarrow u = x_1(\tau)=x_1(T) = x_f .
\end{align*}
The constraints $x_1(0)=s$, $ x_2(0)=0$, $x_1(\tau)=x_f$, $  x_2(\tau)=0$, $u(\tau)=x_f$ give us
\begin{align*}
\left\{\begin{matrix}
C_1 = s \\
C_2+C_3 = 0 \\
C_1\cos \tau+C_2\sin \tau+C_3\tau\sin \tau+C_4 \tau \sin \tau = x_f \\
(C_4-C_1) \sin \tau+(C_2+C_3)\cos \tau-C_3 \tau \sin \tau + C_4 \tau \cos \tau = 0 \\
-2C_3 \sin \tau+2C_4 \cos \tau = x_f
\end{matrix}\right.
\end{align*}
solving which we can obtain the solution of $x_1^c(t)$ and $x_2^c(t)$.

\section{Simulation} \label{sec:simulation}
In this section, we compare our solution with the numerical solution given by \cite{ICLOCS2}.
The following simulations are all done on a personal computer with MATLAB R2020b.
\subsection{Experiment 1: Initial state $(s,0)$ and $0= s<x_f$}
To illustrate the results we obtained in subsection \ref{subsec:s1}, we divide our experiments into two parts: $T$ is small and $T$ is large.
We fix $x_f = 2$.
\subsubsection{$T$ is small}
When the terminal time $T$ is small, we have the following results which satisfy our analytical solution and physical explanation.
Let $T=1$.
Fig. \ref{fig:EX1_T1} shows the numerical solution (NS) using the toolbox \cite{ICLOCS2} and the analytical solution (AS) we derived in Eqn. \eqref{eqn:finalx2}.
The numerical solution, the analytical solution, and the physics analysis coincide with each other.
\begin{figure}[!htp]
\begin{center}
\includegraphics[width=8.4cm]{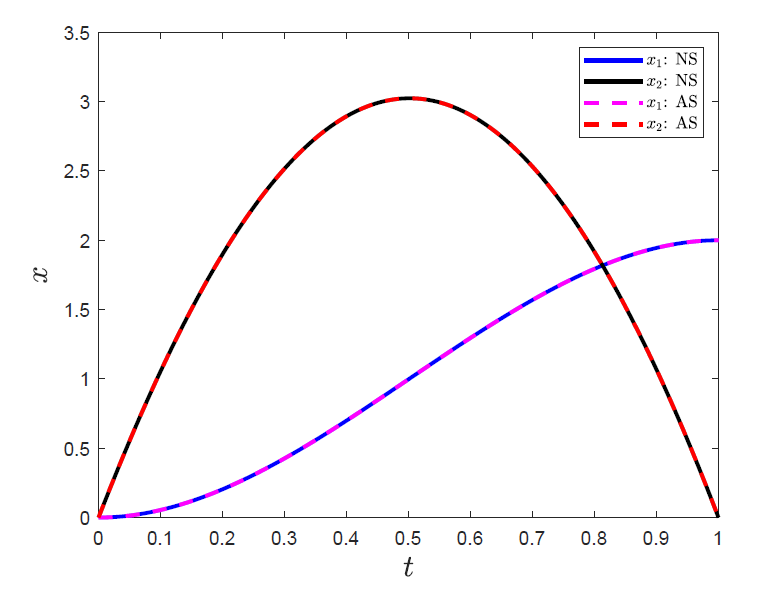}
\caption{Evolution of state: terminal time $T=1$ (legend ``AS'' denotes analytical solution, ``NS'' denotes numerical solution).} 
\label{fig:EX1_T1}
\end{center}
\end{figure}

\subsubsection{$T$ is large}
When the terminal time $T$ is large, say, $T=5$, we have the following results.
Fig. \ref{fig:EX1_T2} shows that our analytical solution coincides with the numerical solution.
Moreover, the optimal cost given by the numerical solution is $J=1.2732 $ which is indeed $\frac{x_f^2}{\pi}$.
\begin{figure}[!htp]
\begin{center}
\includegraphics[width=8.4cm]{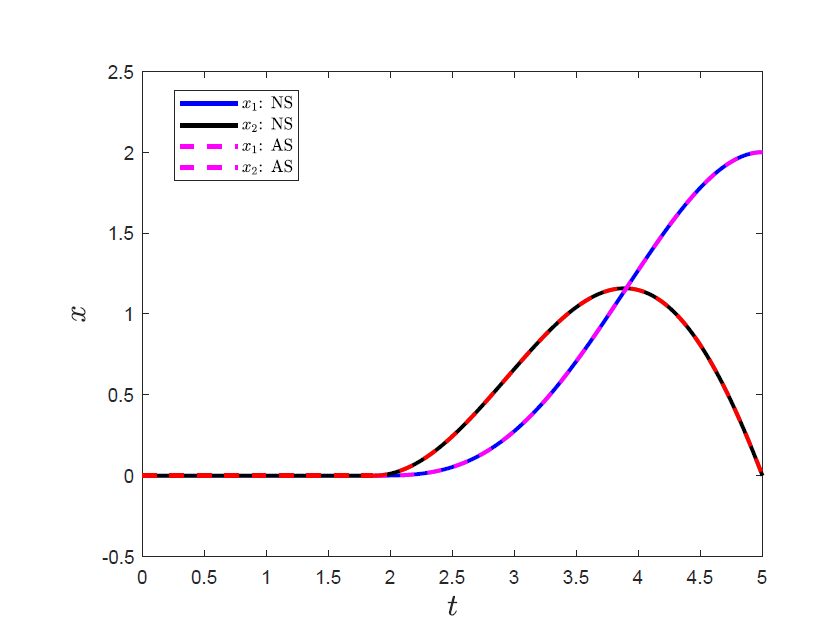}
\caption{Evolution of state (``WM''): terminal time $T=5$.} 
\label{fig:EX1_T2}
\end{center}
\end{figure}

Consider the harmonic oscillator in Fig. \ref{fig:hc}  to be moved from its initial natural state to its final state.
Physically, we would want to pull it right.
If the total time is short, we would pull it right directly to the end position.
If the terminal time is long, we would try to stay at the initial position because it needs less energy to hold the oscillator.
So both simulation and analytical analysis satisfy our physical observations.
\subsection{Experiment 2: initial state $(s,0)$ and $s>0$}
When $T$ is large, we would also want to stay at the beginning because the longer the oscillations, the more energy is needed to hold it at the same place.
Let $s = 1$, $x_f=2$, $T=5$.
Fig. \ref{fig:EX2_T2} shows that our solution is identical to our numerical solution obtained using the toolbox \cite{ICLOCS2}.
The optimal cost is $J\approx 3.918$, $\tau \approx 2.568$.
\begin{figure}[!htp]
\begin{center}
\includegraphics[width=8.4cm]{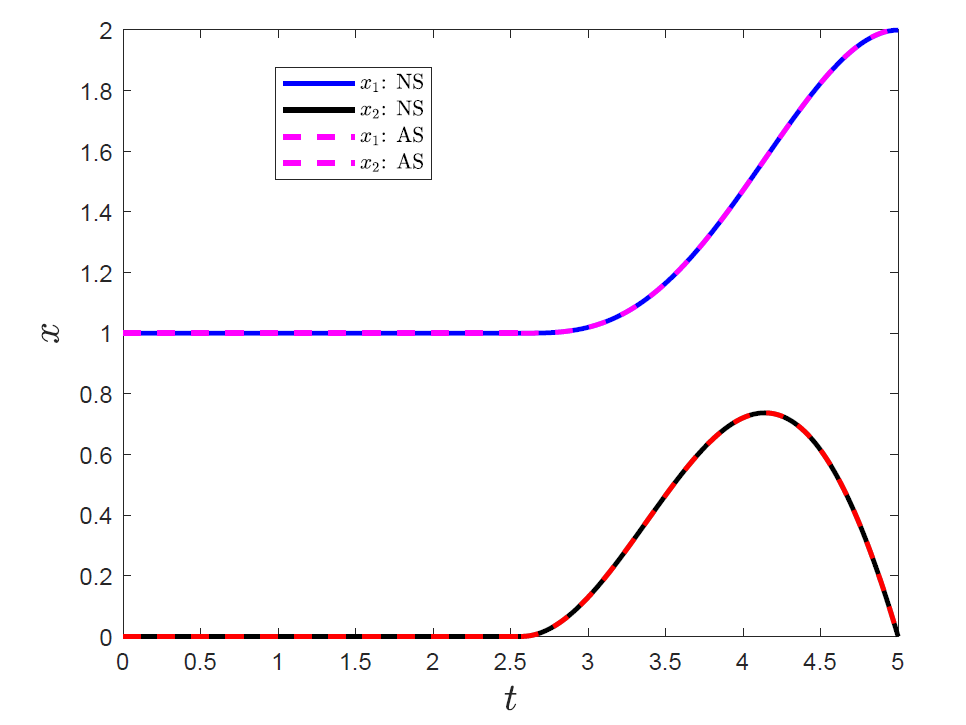}
\caption{Large terminal time (``WM''): $s=1, \; x_f = 2, \; T=5$.} 
\label{fig:EX2_T2}
\end{center}
\end{figure}

\subsection{Experiment 3: initial state $(s,0)$ and $s<0$ and $x_f >0$}
Let $s = -2$, $x_f = 1$.
In this experiment, we find that as $T\rightarrow \infty$, the staying point will be $(0,0)^\top$.
The optimal cost using our method is $J\approx 1.5109$, $\tau_1 = 3.3333$, $\tau_2=5.1890$, $x_{1s}=0.2279$ as shown in Fig. \ref{fig:EX3_T2}.
If $|s|>|x_f|$, the staying point $x_{1s}>0$, otherwise, $x_{1s}<0$.
\begin{figure}[!htp]
\begin{center}
\includegraphics[width=8.4cm]{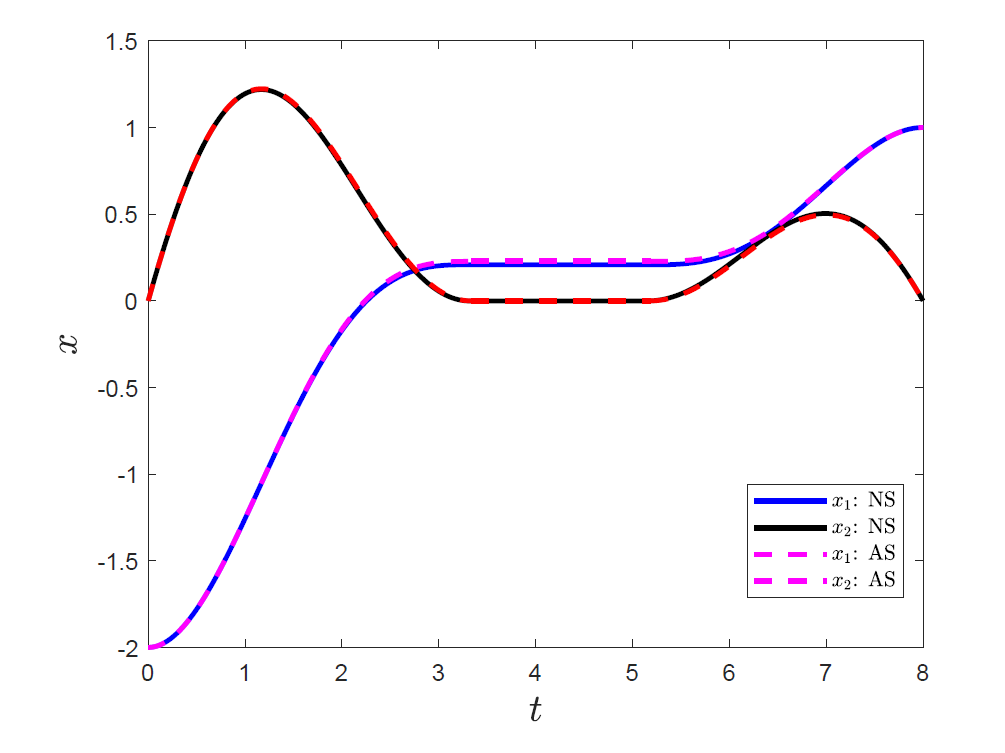}
\caption{Large terminal time (``MWM''): $s=-2, \; x_f=1,\; T=8$.} 
\label{fig:EX3_T2}
\end{center}
\end{figure}

Physically, this scenario means that, at the beginning, the oscillator is compressed, and we want to pull it to a stretched state ($x_f > 0$).
Obviously, it will stay at some position in the middle for the same reason as aforementioned.

\subsection{Experiment 4: initial state $(s,0)$ and $s<x_f<0$}
Let $x_f = -1$, $s=-2$.
This scenario is symmetric to Experiment 1.
We found that $\tau \approx 2.432$, $J \approx 3.918$ using both our method and the toolbox.
Interestingly, it has the same energy consumption compared to its symmetric case.
As we can see from Fig. \ref{fig:EX4_T2}, our solution and the numerical solution given by the toolbox are the same.
\begin{figure}[!htp]
\begin{center}
\includegraphics[width=8.4cm]{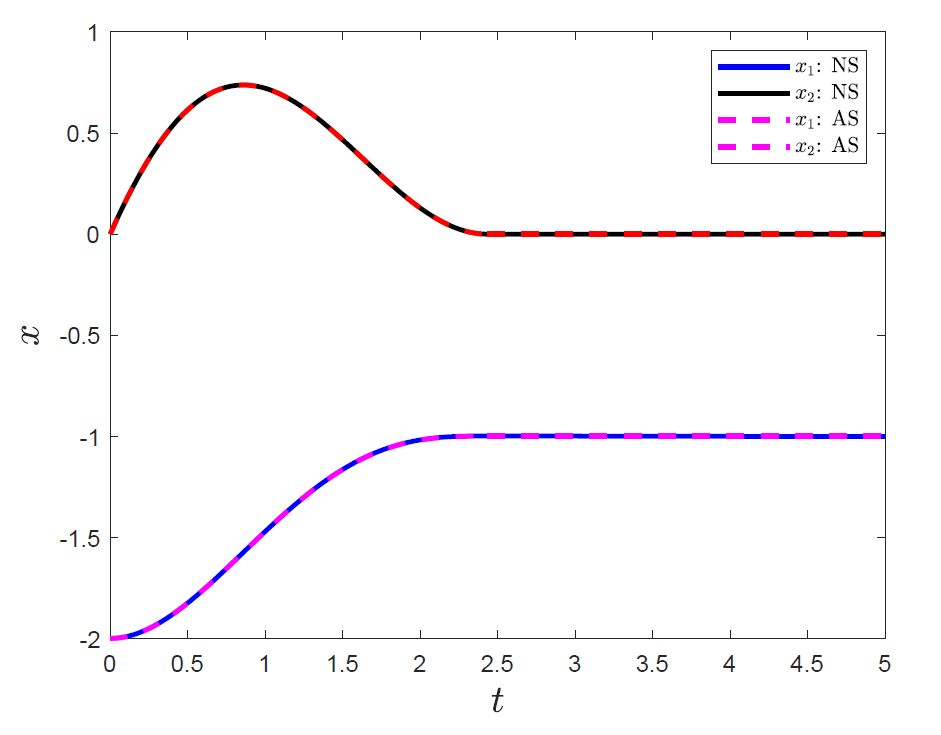}
\caption{Large terminal time (``MW''): $s=-2, \;x_f=-1,\; T=5$.} 
\label{fig:EX4_T2}
\end{center}
\end{figure}
These phenomena coincide from the physics perspective of the optimal control of the harmonic oscillator.
\section{Conclusion}\label{sec:conclusion}
In this paper, we examined and solved the optimal control problem of controlling a harmonic oscillator in a forward motion in terms of minimizing the energy given a fixed terminal time and terminal position.
More specifically, we found the bound of the terminal time when the solutions start becoming non-smooth, and we derived the explicit analytical solution of the optimal controller when the initial state is in the equilibrium of the autonomous system.
We also analyzed the optimal solution when the initial state is in a state of stretching or compression.
Simulation results verified our analysis and we provided physical justification of our theoretical results.
These results shed some light on the optimal swimming policy in a vortex.
We expect this work will also give some insight into other linear time-invariant systems with complex eigenvalues.
Future work will further extend to the unsolved optimal control \cite{MI_ACC} of similar systems with state-dependent switched dynamics or stage cost which will appear multiple switching phenomena at the switching interface.

\bibliographystyle{./IEEEtran} 
\bibliography{./IEEEabrv,./IEEEexample}

\end{document}